\documentclass[11pt]{article}
\usepackage[numbers]{natbib}
\usepackage{amsmath,amsthm,amssymb,amsfonts}
\usepackage[margin=1in,top=1in,bottom=1in]{geometry}

\newcommand{\Esymb}{{\bf E}}

\DeclareMathOperator*{\E}{\Esymb}

\newcommand{\ve}[1]{\mathbf{#1}}
\newcommand{\cdo}{\mathsf{do}}
\newcommand{\kl}{d_{\mathrm{KL}}}
\newcommand{\tv}{d_{\mathrm{TV}}}
\newcommand{\eps}{\varepsilon}
\newcommand{\wh}{\widehat}
\newcommand{\wt}{\widetilde}
\newcommand{\pa}{\mathbf{Pa}}
\newcommand{\Pa}{\mathbf{Pa}}
\newcommand{\poly}{\mathrm{poly}}
\newcommand{\szk}{\mathrm{SZK}}
\newcommand{\bpp}{\mathrm{BPP}}
\newcommand{\distckt}{\textsc{DistCkt}}
\newcommand{\bayesmar}{\textsc{BayesMarginal}}
\newcommand{\circeval}{\textsc{CircEval}}
\newcommand{\bayesmarmult}{\textsc{BayesMarginalMult}}
\newcommand{\all}{\mathrm{all}}
\newcommand{\np}{\mathrm{NP}}
\newcommand{\ber}{\mathrm{Bernoulli}}

\newcommand{\ignore}[1]{}

\renewcommand{\vec}[1]{{\bf{#1}}}

\newtheorem{theorem}{Theorem}[section]
\newtheorem{claim}[theorem]{Claim}
\newtheorem{lemma}[theorem]{Lemma}
\newtheorem{fact}[theorem]{Fact}
\newtheorem{definition}[theorem]{Definition}
\newtheorem{remark}[theorem]{Remark}
\newtheorem{assumption}[theorem]{Assumption}

\newcommand{\cP}{\mathcal P}

\usepackage{enumitem}
\setitemize{noitemsep,topsep=0pt,parsep=0pt,partopsep=0pt}
\setlist[enumerate]{leftmargin=*}

\usepackage[utf8]{inputenc} 
\usepackage[T1]{fontenc}    
\usepackage{hyperref}       
\usepackage{url}            
\usepackage{booktabs}       
\usepackage{amsfonts}       
\usepackage{nicefrac}       
\usepackage{microtype}      
\usepackage{xcolor}         
\usepackage{caption}
\usepackage{subcaption}
\usepackage{graphicx}
\usepackage{tikz}
\usepackage{tkz-tab}
\usepackage{bm}

\usepackage[ruled,vlined]{algorithm2e}
\title{Efficient inference of interventional distributions}

\author{
Arnab Bhattacharyya\\
National University of Singapore\\
\texttt{arnabb@nus.edu.sg}
\and
Sutanu Gayen\\
National University of Singapore\\
\texttt{sutanugayen@gmail.com}
\and
Saravanan Kandasamy\\
Cornell University\\
\texttt{sk3277@cornell.edu}
\and
Vedant Raval\\
Indian Institute of Technology Delhi\\
\texttt{skved004@gmail.com}
\and
N. V. Vinodchandran\\
University of Nebraska-Lincoln\\
\texttt{vinod@cse.unl.edu}
}

\date{May 2021}

\allowdisplaybreaks

\begin{document}
\maketitle
\begin{abstract}
We consider the problem of efficiently inferring interventional distributions in a causal Bayesian network from a finite number of observations. Let $\mathcal{P}$ be a causal model on a set $\ve{V}$ of  observable variables on a given causal graph $G$. For sets $\vec{X},\vec{Y}\subseteq \ve{V}$, and setting ${\bf x}$ to $\vec{X}$, let $P_{\bf x}(\vec{Y})$ denote the interventional distribution on $\vec{Y}$ with respect to an intervention ${\bf x}$ to  variables $\vec{X}$. Shpitser and Pearl (AAAI 2006), building on the work of Tian and Pearl (AAAI 2001), gave an exact characterization of the class of causal graphs for which the interventional distribution $P_{\bf x}({\vec{Y}})$ can be uniquely determined. 

We give the first efficient version of the Shpitser-Pearl algorithm. In particular, under natural assumptions, we give a polynomial-time algorithm that on input a causal graph $G$  on observable variables ${\ve V}$, a setting ${\bf x}$ of a set $\vec{X} \subseteq {\ve V}$ of bounded size, outputs succinct descriptions of both an evaluator and a generator for a distribution $\hat{P}$ that is $\varepsilon$-close (in total variation distance) to $P_{\bf x}({\vec{Y}})$ where $Y=\ve{V}\setminus \vec{X}$, if $P_{\bf x}(\vec{Y})$ is identifiable. 

We also show that when $\vec{Y}$ is an arbitrary set, there is no efficient algorithm that outputs an evaluator of a distribution that is $\varepsilon$-close to $P_{\bf x}({\vec{Y}})$ unless all problems that have statistical zero-knowledge proofs, including the Graph Isomorphism problem,  have efficient randomized algorithms. 
\end{abstract}

\section{Introduction}
{\em Density estimation} and {\em parameter learning} are old and classical problems in statistics, studied since the field's inception (e.g., \cite{DG85, Scott15, Silverman18}). A more recent focus has been on designing {\em computationally efficient} learning algorithms, especially in high-dimensional settings. The seminal work of Kearns, Mansour, Ron, Rubinfeld, Schapire, and Sellie \cite{KMRRSS94} set forth the two core computational requirements for distribution learning: (i)  {evaluation}, and (ii) {sampling}/{generation}. An approximate evaluator for a distribution $P$ takes a domain element $\vec{v}$ and outputs the mass of $\hat{P}$ at $\vec{v}$, where $\hat{P}$ is a distribution close to $P$. Similarly, an approximate generator for $P$ takes as input a random seed and outputs an element $\vec{v}$ distributed according to $\hat{P}$. It is desirable to learn both (approximate) evaluator and generator representations of a distribution, as they can be useful for different downstream inference tasks.

Distribution learning is especially interesting when we cannot directly access the distribution to be learned. One of the most prominent example of such a situation arises in {\em causal effect estimation}. To estimate the effect of a treatment intervention on some outcome, the ``gold standard'' is to conduct a randomized experiment where a random sub-population is forcibly treated. However, in many settings (e.g., medicine, economics), it is often not feasible to conduct such experiments, and the only recourse is to make use of observational data. In this work, we study the problem of efficiently learning evaluators and generators for interventional distributions from observational samples.  

Our discussion of causality will be in Pearl's language of {\em causal Bayesian networks} (or {\em causal Bayes nets} for short) \cite{Pearl09}. A causal Bayes net is a standard Bayes net that is reinterpreted causally. Specifically, it makes the assumption of {\em modularity}:  for any variable $X$, the dependence of $X$ on its parents is an autonomous mechanism that does not change even if other parts of the network are changed. This allows assessment of external interventions, such as those encountered in policy analysis, treatment management, and planning. 

We now fix some basic notation for causal Bayes nets to ground the subsequent discussion.
The underlying structure of causal Bayes net $\cP$ is a directed acyclic graph $G$. The graph $G$ consists of $n+h$ nodes where $n$ nodes correspond to the {\em observable} variables $\vec{V}$ while the $h$ additional nodes correspond to $h$ {\em hidden variables} $\vec{U}$. We assume that the observable variables take values over a finite alphabet $\Sigma$. Defining the key notion of {\em c-components} is deferred to Section \ref{sec:prelims}.

The observational distribution $P$ on $\vec{V}$ is obtained by interpreting $\cP$ as a standard Bayes net over $\vec{V} \cup \vec{U}$ and marginalizing to $\vec{V}$.
An {\em intervention} is specified by a subset $\vec{X} \subseteq \vec{V}$ of variables and an assignment\footnote{Consistent with the convention in the causality literature, we will use a lower case letter (e.g., $\vec{x}$) to denote an assignment to the subset of variables corresponding to its upper case counterpart (e.g., $\vec{X}$).} $\vec{x} \in \Sigma^{|\vec{X}|}$. In the  interventional distribution, the variables $\vec{X}$ are fixed to $\vec{x}$, while each variable $W \in (\vec{V} \cup \vec{U})\setminus \vec{X}$ is sampled as it would have been in the original Bayes net, according to the conditional distribution $W \mid \Pa(W)$, where $\Pa(W)$ (parents of $W$) consist of either variables previously sampled in the topological order of $G$ or variables in $\vec{X}$ set by the intervention. The marginal of the resulting distribution to $\vec{V}$ is the interventional distribution denoted by $P_{\vec{x}}$. 

The problem of interest for us is to efficiently learn $P_{\vec{x}}$, given samples from $P$. If we leave aside considerations of efficiency, this problem was solved in the important prior works of Shpitser and Pearl \cite{SP06} and, independently, Huang and Valtorta \cite{HV06}. We focus here on the work of Shpitser and Pearl.  They characterized the set of graphs $G$ for which $P_{\vec{x}}$ is not statistically identifiable from $P$. Moreover when $P_{\vec{x}}$ is identifiable, and assuming $P$ is strictly positive, they gave an algorithm that explicitly generates a formula for $P_{\vec{x}}$ in terms of $P$. 

If one only has access to samples from $P$ though, Shpitser and Pearl's algorithm does not guarantee that $P_{\vec{x}}$ is learned up to bounded error in total variation distance, either as an evaluator or as a generator, using a polynomial number of samples and time. This is the guarantee which we achieve, when the graph $G$ has bounded in-degree and bounded c-components, the set $\vec{X}$ is bounded, and there is no marginalization. 

\begin{theorem}[Informal]\label{thm:main1inf}
Let $\cP$ be a causal Bayes net where the underlying DAG $G$ has in-degree at most $d$ and c-component size at most $k$. There is an algorithm that given such a causal Bayes net $G$, a subset $\vec{X} \subset \vec{V}$ of size at most $\ell$ such that $P_{\vec{x}}(\vec{Y})$ is identifiable from $P$ where 
$\vec{Y} = \vec{V}\setminus \vec{X}$, and $\eps$; with constant probability, outputs (learns) a distribution $\hat{P}$ as an approximate evaluator and generator 
so that 
$\tv({P}_{\vec{x}}(\vec{Y}),\hat{P}) \leq \eps$.   The algorithm uses $m=\tilde{O}\left(\frac{nk^{O(k)}|\Sigma|^{O(k\ell+kd)}}{\alpha^{k\ell}\eps^2}\right)$ samples and $O(m(n+|\Sigma|^{kd+d}))$ time, where $\alpha$ is the minimum probability for any non-empty event defined by the variables in c-components intersecting $\vec{X}$ and such c-components' parents. 
\end{theorem}

Note that unlike in \cite{SP06}, we do not require $P$ to be a strictly positive distribution. Indeed if we let $\alpha$ above be a lower bound on $P(\vec{v})$, then $\alpha \leq |\Sigma|^{-n}$, and hence, the above sample/time complexity would be exponential. So, it is crucial for efficiency that $\alpha$ be a lower bound on the probability of events defined by only a bounded number of variables.

Theorem \ref{thm:main1inf} is proved by showing how to implement the Shpitser-Pearl algorithm while guaranteeing the bound on the distance between the learned and true interventional distributions. This analysis relies on recent work \cite{bgmv20, chowliu} which examined fixed-structure learning of discrete Bayes nets without hidden variables. However, unlike the `chain-like' factorization $P(\vec{v}) = \prod_i P(v_i \mid \vec{v}_{\Pa(i)})$ of the probability mass function for a Bayes net, the interventional probability distribution may not admit a nice factorization; see the examples later in Section \ref{sec:prelims}. Our algorithm exposes a decomposition of $P_{\vec{x}}$ into a product of interventional distributions that are either on a small sample space or have `chain-like' factorization, and we show both can be learned efficiently using samples from $P$.  The sample complexity of Theorem~\ref{thm:main1inf} is nearly optimal in terms of $n$ and $\eps$ but it remains an interesting open problem to improve the dependence on $d, k, \ell, \alpha$ and $ \Sigma$.

One drawback of Theorem \ref{thm:main1inf} is that it considers the interventional distribution on $\vec{V}\setminus \vec{X}$, whereas the causal effect on only a subset  $\vec{Y}$ of outcome variables may be relevant and is considered by Shpitzer-Pearl~\cite{SP06}. We show that unfortunately, we can not hope for an analog of Theorem \ref{thm:main1inf} for an arbitrary subset $\vec{Y}$. In particular,  we show that the existence of a learning algorithm that outputs an approximate evaluator representation of $P_{\vec{x}}(\vec{Y})$ for an arbitrary $\vec{Y} \subseteq \vec{V}$ will lead to  efficient randomized algorithms for all problems in the complexity class $\szk$ which contains hard problems such as the Graph Isomorphism problem. 
Indeed, we establish the hardness even when the intervention set is empty, i.e., when the goal is to learn the marginal on a subset of variables and when the Bayes net has indegree at most 2.

\begin{theorem}[Informal]\label{thm:main2inf} Suppose there is a randomized polynomial-time algorithm that on input a Bayes net (without hidden variables) distribution $P$ on $\vec{V}$,  $\vec{Y} \subseteq \ve{V}$, and $\eps$; outputs a representation $r$ of a distribution $\hat{P}$ so that (1) $\tv(P(\vec{Y}),\hat{P}) \leq \eps$, and (2) for every $\vec{y}$, $\hat{P}(\vec{y})$ can be evaluated (or even multiplicatively $(1\pm\eps)$-approximated) efficiently using $r$.  Then all problems that have statistical zero knowledge (the complexity class $\szk$), including the Graph Isomorphism problem, can be solved in  randomized polynomial time.  \end{theorem}

Thus the problem of outputting an  evaluator representation of a distribution that approximates the effect of an interventional distribution on an arbitrary subset, even when it is  identifiable in the sense of Shpitser-Pearl algorithm, is computationally hard.  The formal statement of this theorem along with the proof is presented in Section~\ref{sec:hard}.

\begin{remark}
{\em
Using notations of Theorem~\ref{thm:main1inf}, the lower bound of Theorem~\ref{thm:main2inf} corresponds to the case $d=2,k=1,l=0$ and $\alpha=1$. We also note that for product distributions on $n$-variables, learning such an approximate evaluator for marginals is trivial.
}
\end{remark}

\subsection{Related Work}
Pearl~\cite{pearl89} introduced Causal Bayesian Networks to formally define interventions and causal effects. Tian and Pearl~\cite{tian2002general} first studied the problem of identification of causal effects from observations, and gave graphical characterizations for such identifications when the intervention is atomic, i.e. consisting of a single variable and we are interested to determine the effect of this variable on the rest of the variables. They also gave an expression for such an atomic intervention in terms of observational quantities whenever the later is identifiable. Subsequently, Shpitser and Pearl~\cite{SP06} and Huang and Valtorta~\cite{HV06} independently generalized~\cite{tian2002general} for non-atmoic interventions on any subset of variables, giving graphical characterizations for identifiability and an expression of the causal effect whenever it is identifiable. Recently Barenboim et al.~\cite{bareinboim} gave estimators for causal effects using weighted regression.

Finite sample bounds for causal inference are being studied only recently. Acharya et al.~\cite{acharya2018learning} gave finite bounds for goodness-of-fit testing for non-parametric causal models using observation and experimental data. Bhattacharya et al.~\cite{bhattacharyya2020learning} gave finite bounds for learning the effect of atomic interventions on the rest of the graph using observational samples, giving finite bounds for~\cite{tian2002general}. The finite sample and time bounds shown in this paper can be seen as a generalization of~\cite{bhattacharyya2020learning} for learning non-atomic interventions on the rest of the graph using observational samples. Our work can be seen as a finite-sample version of~\cite{SP06}.

Testing various properties of polynomial-time samplable distributions such as uniformity, entropy, and closeness is a fundamental problem that characterizes several zero-knowledge complexity classes~\cite{watson2016complexity,SV03,malka2015achieve}. We show our hardness result employing a connection between testing polynomial-time samplable distributions and testing Bayesian networks. Bayesian network is an important statistical model for which finite sample bounds for testing and learning have been given recently~\cite{bgmv20,cdks,chowliu}.

\section{Preliminaries} \label{sec:prelims}

\begin{figure}
\begin{algorithm}[H]
\SetKwInOut{Input}{Input}
\SetKwInOut{Output}{Output}
\Input{   Assignments $\ve{x},\ve{y}$, Observational distribution $P$, ADMG $G$}
\Output{   $P_{\ve{x}}(\ve{y})$}
\nl \If{$\ve{x} = \emptyset$}
{\Return $\sum_{\ve{v \setminus y}} P(\ve{v})$}

\nl \If{$\ve{V} \setminus \mathrm{An}(\ve{Y})_{G} \neq \emptyset$}
{\Return ID($\ve{y}, \ve{x} \cap \mathrm{An}(\ve{Y})_{G}, 
\sum_{\ve{v} \setminus \mathrm{An}(\ve {Y})_{G}} P , \mathrm{An}(\ve{Y})_{G}$).}

\nl \If{$\vec{W} := (\vec{V \setminus X}) \setminus \mathrm{An}(\ve{Y})_{G_{\overline{\ve{X}}}}$ is non-empty}{For arbitrary assignment $\ve{w}$, \Return  ID($\ve{y}, \ve{x \cup w} , P, G$).}

\nl \If{$C(G \setminus \ve{X} = \{ \ve{S}_1, \ldots, \ve{S}_k \}$)}
{\Return $\sum_{\ve{v} \setminus (\ve{y \cup x})} \prod_{i} \text{ID(} s_i, \ve{v} \setminus \ve{s}_i \text{)}, P, G $.}

\nl \If{$C(G \setminus \ve{X}) = \{ \ve{S} \}$ is singleton}
{
a) \If{$C(G) = G$}{\Return FAIL.}
b) \If{$\ve{S} \in C(G)$} {\Return $\sum_{\ve{s \setminus y}} \prod_{i | V_i \in \ve {S}} P(v_i \mid v_1, v_2, \ldots, v_{i-1})$.}
c) \If{$\exists \ve{S}^{\prime}: \ve{S} \subset \ve{S}^{\prime} \in C(G) $}
{\Return ID($\vec{y}, \vec{x}\cap \vec{S}^{\prime}, \prod_{i | V_i \in \ve{S}^{\prime}}
P(V_i \mid (V_1,V_2,\ldots,V_{i-1}) \cap \ve{S}^{\prime}, (v_1,\ldots,v_{i-1})\setminus \ve{S}^{\prime}), \ve{S}^{\prime} $ ).}
}
\caption{ID($\ve{y,x},P,G$)}
\label{algo:id}
\end{algorithm}
\end{figure}

\ignore{
Some of the formal definitions and notation used in the paper are discussed in Appendix~\ref{sec:mprelims}.  In Appendix~\ref{sec:id}, we provided a detailed explanation of the identification algorithm (Algorithm~\ref{algo:id}) of \cite{SP06}.  Our hardness result is discussed in Appendix~\ref{sec:hard}, where we show that the problem of outputing an evaluator/sampler representation of a distribution that approximates the effect of an interventional distribution on an arbitrary subset is computationally hard.
}

\paragraph{Notation.} We use bold and non-bold fonts to denote sets and singleton variables respectively. We use capital and small letters to denote variables and values taken by them respectively. 

For two distributions $P$ and $Q$ over $\Omega$, their total variation distance $\tv(P,Q)=\frac{1}{2}\sum_{i\in \Omega}|P_i-Q_i|$ and their KL distance $\kl(P,Q)=\sum_{i\in \Omega} P_i\ln \frac{P_i}{Q_i}$. Pinsker's inequality says $\tv(P,Q)\le \sqrt{0.5 \cdot \kl(P,Q)}$.

\begin{definition}
A Bayesian network (or Bayes net in short) $P$ is a distribution over a set of variables $\vec{V}$ defined with respect to a directed acyclic graph $G$. Each $v\in \vec{V}$ takes values from an alphabet $\vec{\Sigma}$. $G$ encodes dependence relationship between the variables: every variable is independent of its non-descendants conditioned on its immediate parents. Hence according to the Bayes rule, $P$ factorizes in the topological order of $G$ such that each $v\in V$ is only conditioned on its parents: $P[\vec{V}]=\prod_i P[V_i \mid \vec{V_{\mathrm{parents}(i)}]}$.
\end{definition}

\paragraph{Causality.}
We follow Pearl's formalism of Causal Bayesian Networks~\cite{Pearl09}.

\begin{definition}[Causal Bayes Net] \label{def:causal Bayesnet}
A Causal Bayes net $P$ is a Bayes net over the variables $\vec{V} \cup \vec{U}$ defined with respect to a DAG over $\vec{V}\cup \vec{U}$. In this DAG, each edge $(X\rightarrow Y)$ denotes a {\em directed} causal relationship from $X$ to $Y$. Only the variables in $\ve{V}$ (observed variables) but none from $\ve{U}$ (hidden variables) can be observed in the samples from $P$.

Given a $X\subseteq V$, and an assignment $x$ to $X$, the interventional distribution $P_x(V)$, is defined on the DAG where the set of incoming edges to $X$ are removed and $X$ is fixed to $x$ with probability 1. All other variables follow the usual parent-child relation and factorization of $P$.
\end{definition}

It is a standard to assume that the unobservable variables in $\vec{U}$ have exactly two observable children.  
Causal Bayes nets are often represented using a graphical representation over only the observed vertices $\vec{V}$.
An Acyclic Directed Mixed Graph (ADMG in short) is such a representation, which consists of a set of variables and two kinds of edges: directed edges $E^{\rightarrow}$ and bi-directed edges $E^{\leftrightarrow}$. The directed edges $(X\rightarrow Y)$ denote parent-child relationship as in a DAG. The bi-directed edges $(X\longleftrightarrow Y)$ denote an indirect correlation between $X$ and $Y$ through a sequence of hidden variables (from $U$) in between. Thus, the edge set of an ADMG is the union of the directed edges $E^{\to}$ and the bidirected edges $E^{\leftrightarrow}$.  Also for any given set $\vec{S} \subseteq \vec{V}$, $G_{\overline{\vec{S}}}$ denotes the graph obtained from $G$ by removing the incoming edges to $\vec{S}$.

We use Structural Causal Models (SCMs) as defined in \cite{Pearl09}.  
SCMs are defined on Acyclic Directed Mixed Graphs (ADMGs) which contain (i) directed edges between observable variables that form a directed acyclic graph; and (ii) bi-directed edges between observable variables that represent an unobservable confounder between the two variables.  For any subset $\ve{S} \subseteq \ve{V}$, $\Pa^{+}(\ve{S})$ denotes the set of all observable parents of $\ve{S}$ (including $\ve{S}$). Let $\pa^{-}(\ve{S})=\pa^{+}(\ve{S})\setminus \ve{S}$.

For ADMGs, the c-component is a key notion that plays a crucial role in the identification of causal effects.

\begin{definition}[c-component]
An ADMG (or a subgraph of an ADMG) $G$ is a {\em c-component} if all the vertices of $G$ are connected by bi-directed edges, or in other words, there exists a subset of bi-directed edges that form a spanning tree of $G$.  
\end{definition}

\begin{definition}[c-component factorization] \label{defn:cfact}
For any ADMG $G$ over observables $\vec{V}$, the {\em c-component factorization} is the partition of vertices $\vec{V} = \{ \vec{C_1,C_2,\ldots,C_{m}} \}$ where the induced subgraphs $G[\ve{C_i]}$s are (maximal) c-components.
\end{definition}

\begin{definition}[Effective parents of $V_i$ in $G$]
For any vertex $V_i$, let $\vec{C}$ be the c-component of $G$ that contains $V_i$.  Then the {\em effective parents} of $V_i$ is $\Pa^{+}(\vec{C}) \cap \{V_1,V_2,\ldots,V_{i-1}\}$.
\end{definition}

\begin{definition}[$\alpha$-strong positivity for c-components]\label{ass:bal}
A c-component $C$ is said to be $\alpha$-strongly positive if for every assignment $\ve{z}$, $P(\ve{\pa^{+}(C)}=\ve{z}) \ge \alpha$. 
\end{definition}

We now discuss the identification algorithm of \cite{SP06} as presented in Algorithm~\ref{algo:id}.  In the remainder of this section, we will go through the different steps of the algorithm by considering two examples.  

\subsection{Identification Revisited} \label{sec:id}
In this section we discuss the identification algorithm
\cite{tianthesis} in detail to provide the intuition behind it.  Suppose we are given an ADMG $G$ and an observational distribution $\Pr[\ve{V}]$.  Let $\vec{X,Y} \subseteq \vec{V}$ be disjoint subsets and $\ve{x},\ve{y}$ be assignments to $\vec{X}$ and $\vec{Y}$.  The goal of the {\em identification} question is to determine the probability $\Pr[\ve{y} \mid do(\ve{x})]$.  Here we consider the version of the algorithm Algorithm~\ref{algo:id} mentioned in \cite{SP06}.

When we are interested in a query of the form $\Pr[\ve{y} \mid do(\ve{x})]$, certain trivial cases can be handled easily and those cases will correspond to the base cases of their algorithm.  The first three steps below will describe those cases.

\begin{enumerate}
\item Step 1 handles the base case when $\ve{X}$ is an empty set.  In this case, the algorithm directly outputs the observational probability $\Pr[\ve{y}]$.

\item Step 2 handles the base case when $\Theta:=\ve{V} \setminus \mathrm{An}(\ve{Y})_{G}$ is non-empty.  It is clear by definition that the vertices of $\Theta$ is not affecting $\ve{Y}$.  Therefore, $\Theta$ can be removed from the original graph $G$.  

\item Step 3 handles the base case when $\ve{W} := (\ve{V} \setminus \ve{X}) \setminus \mathrm{An}(\ve{Y})_{G_{\overline{\ve{X}}}}$ is non-empty.  Suppose $\ve{W}$ is non-empty.  Since $\ve{W}$ does not affect $\ve{Y}$ once the vertices $\ve{X}$ are intervened, the effect of $do(\ve{X})$ on $\ve{Y}$ is equivalent to the effect of $do(\ve{W \cup X})$ on $\ve{Y}$.  Hence it is convenient to state the question apriori with $\ve{W}$ being a subset of $\ve{X}$.  This step makes it easier to factorize the required query (based on c-component factorization, Step 4). 

\item Suppose $\ve{S}_1$, \ldots, $\ve{S}_k$ (for $k>1$) are the c-components of $G[\ve{V \setminus X}]$. Then we get the following formula:
\begin{align*}
\Pr[\ve{y} \mid do(\ve{x})] = \sum_{\ve{v} \setminus (\ve{y} \cup \ve{x})} \prod_{i \in [k]} \Pr[\ve{s}_i \mid do(\ve{v}\setminus \ve{s}_i)],
\end{align*}
where the product term is the resultant of the well-known c-component factorization formula (Definition~\ref{defn:cfact}) and the summation is the marginalization operation.

\item Suppose the c-component of $G[\ve{V} \setminus \ve{X}]$ is a singleton $\ve{S}=\{\vec{V \setminus X}\}$.  Here we have three interesting special cases.
\begin{enumerate}
\item Consider the case where the c-component of $G$ is $G$ itself.  This results in the existence of hedge (with root set $\ve{S}$ and internal nodes $\ve{X}$, and some edges can be removed to make sure each internal node has exactly one outgoing edge.  Please refer \cite{SP06} for the precise definition of Hedge).  Hence, it is impossible to uniquely determine the required query.
\item Consider the case where $\ve{S}$ is a c-component of $G$.  This means there is no-bidirected edge between $\ve{S}$ and $\ve{V \setminus S}$.  Due to the absence of bi-directed edges, no backdoor paths exist from $\ve{X}$ to $\ve{S}$, resulting in the following formula:
\begin{align*}
\Pr[\ve{y} \mid do(\ve{x})]
= \sum_{\ve{v}\setminus \ve{y}} \prod_{i \mid V_i \in \ve{S}} \Pr[v_i \mid v_1, \ldots, v_{i-1}]
\end{align*}
whose correctness can be easily verified by Bayes rule and marginalization.  (Note: One can also apply conditional independence properties of the model on top of the above expression which would remove some of the $v_j$'s in the conditional terms.)

\item This is the final case.  Suppose there exists $\ve{S}^{\prime} : \ve{S} \subset \ve{S}^{\prime}$ which is a c-component of $G$.  Here, similar to Step 5(b), since there is no bi-directed edge between $\ve{X} \setminus \ve{S}^{\prime}$ and the rest of the vertices, it is possible to identify the distribution obtained after the intervention $do(\ve{x} \setminus \ve{S}^{\prime})$ - using the following formula:
$$
\Pr[\ve{S}^{\prime} \mid do(\vec{x} \setminus \vec{S^{\prime}})] = \prod_{V_i \in \ve{S}^{\prime}} \Pr[V_i | (\{V_1,\ldots,V_{i-1}\}) \cap \ve{S}^{\prime}, \ve{x} \setminus \ve{s}^{\prime} ].$$
where the assignment $\ve{x} \setminus \ve{s}^{\prime}$ is the one that's consistent with $\ve{x}$.

The idea here is to first determine $do(\ve{x} \setminus \ve{S}^{\prime})$ and then try to recursively determine  $do(\ve{x} \cap \ve{S}^{\prime})$ on top of that -- thus obtaining $do(\ve{x})$.  Hence the required query is obtained by determining $\Pr[\ve{y} \mid do(\ve{x} \cap \ve{S}^{\prime})]$ in the graph $G[\ve{S}^{\prime}]$ with observational distribution $\Pr[\ve{S}^{\prime} \mid do(\vec{x} \setminus \vec{S^{\prime}})]$ defined above.

\end{enumerate}
\end{enumerate}

\ignore{

\begin{figure}
    \begin{subfigure}[b]{0.45\textwidth}
    \centering
    \includegraphics[scale=0.2]{fig1.png}
    \caption{Graph 1}
    \label{fig1}
    \end{subfigure}
    \begin{subfigure}[b]{0.45\textwidth}
    \centering
    \includegraphics[scale=0.2]{fig2.png}
    \caption{Graph 2}
    \label{fig2}
    \end{subfigure}
    \begin{subfigure}[b]{0.45\textwidth}
        \centering
        \includegraphics[scale=0.2]{fig3.png}
        \caption{Graph 3}
        \label{fig3}
    \end{subfigure}
    \begin{subfigure}[b]{0.45\textwidth}
        \centering
        \includegraphics[scale=0.2]{fig4.png}
        \caption{Graph 4}
        \label{fig4}
    \end{subfigure}
    \begin{subfigure}[b]{0.45\textwidth}
        \centering
        \includegraphics[scale=0.2]{fig5.png}
        \caption{Graph 5}
        \label{fig5}
    \end{subfigure}
\end{figure}}

\textbf{Example 1.}  Consider the ADMG shown in Figure~\ref{fig:admg4} and the identification of $P(y,z_1,z_2 \mid do(x))$ from the observational distribution on this graph.  Note that the conditions in steps~1,2 and 3 of Algorithm~\ref{algo:id} do not hold and hence step 4 gets invoked. Since the c-components of $G[\{ Z_1,Z_2,Y \} \setminus \{X\}]$  are $\{Y,Z_1\}$ and $Z_2$, we get the following formula:  $P(y,z_1,z_2 \mid do(x)) = P[y,z_1 \mid do(x,z_2)] \cdot P[z_2 \mid do(x,y,z_1)] $ .

For $P[y,z_1 \mid do(x,z_2)]$, the recursive procedure invokes Step 5b to obtain the following equivalent expression $P[z_1 \mid x ] \cdot P[  y \mid z_1,z_2,x] $.

For the other term $P[ z_2 \mid do(x,y,z_1) ]$, since $Y$ is not an ancestor of $Z_2$, the intervention on $Y$ is not particularly useful. This simplification is taken care by Step 2, which reduces to the following question $P[ z_2 \mid do(x,y,z_1) ] = P^{\prime}[ z_2 \mid do(x,z_1)]$ on Graph~\ref{fig:admg5}, where $P^{\prime}[X,Z_1,Z_2] = \sum_y P[X,Z_1,Z_2,y]$.  In the next recursive step, the algorithm invokes Step 5c with $\vec{S}^{\prime} = \{ Z_2,X \} $, henceforth reduces to a new question $P^{\prime}[ z_2 \mid do(x,z_1)] = P^{\prime \prime}[ z_2 \mid do(x)]$ in Graph~\ref{fig:admg6}, where $P^{\prime \prime}[Z_2,X] = P^{\prime}[X] P^{\prime}[Z_2\mid X,z_1]$.  Finally the algorithm invokes Step 2, obtaining $P^{\prime \prime}[z_2\mid do(x)] = P^{\prime \prime } [ z_2] = \sum_{x} P^{\prime}[x] P^{\prime}[z_2\mid x,z_1] $.

\textbf{Example 2} Now consider the question of identifying $P[y \mid do(x,r,y)]$ on the ADMG shown in Figure~\ref{fig:admg1}.  Note that the conditions in steps 1-4 do not hold for the required query, hence the algorithm directly invokes step 5c where $\ve{S}^{\prime} = \{X,Y,W\}$.  This invocation boils down to the identification question of determining $P^{\prime}[ y \mid do(w,x) ]$ in Graph~\ref{fig:admg2}, where $P^{\prime}[W,X,Y] := P[W,X,Y \mid do(r)] = P[W] P[X\mid W,r] P[Y\mid W,X,r] $.  In the next recursive call, the algorithm will stop at Step 2 resulting in the identification of $P^{\prime \prime}[y \mid do(x)] $ in Graph~\ref{fig:admg3}, where $P^{\prime \prime} = \sum_{w} P^{\prime}[ X,Y,w ] = \sum_w P[w] P[X\mid w,r] P[Y \mid X,w,r]  $.  Finally, Step 5b gets invoked in the next call thus resulting in the following expression:  $$P^{\prime \prime} [ y \mid do(x)]  = P^{\prime \prime } [y \mid x] 
= \dfrac{\sum_{w} P[w] P[x\mid w,r] P[y \mid x,w,r] }{\sum_{w,y} P[w] P[x\mid w,r] P[y \mid x,w,r] }.
$$

\begin{figure}
     \centering
     \begin{subfigure}{0.3\textwidth}
         \centering
        \resizebox{\linewidth}{!}{
            \begin{tikzpicture}
            
                \filldraw [black] (-2,0) circle (1pt) node[below] {\scriptsize $\text{Y}$};
                \filldraw [black] (2,0) circle (1pt) node[below] {\scriptsize $\text{$Z_2$}$};
                \filldraw [black] (0,2) circle (1pt) node[left] {\scriptsize $\text{X}$};
                \filldraw [black] (1,1) circle (1pt) node[below] {\scriptsize $\text{$Z_1$}$};

                \draw [stealth-, semithick] (-2,0) -- (2,0);
                \draw [stealth-, semithick] (-2,0) -- (0,2);
                \draw [stealth-, semithick] (-2,0) -- (1,1);
                \draw [stealth-, semithick] (1,1) -- (0,2);
                \draw [stealth-, semithick] (2,0) -- (1,1);
                
                \draw [dotted, semithick, <->]   (0,2) to[out=0,in=90] (2,0);
                \draw [dotted, semithick, <->]   (1,1) to[out=180,in=40] (-2,0);
                
            \end{tikzpicture}
        }         
        \caption{}
        \label{fig:admg4}
     \end{subfigure}
     \hfill
     \begin{subfigure}{0.3\textwidth}
         \centering
        \resizebox{\linewidth}{!}{
            \begin{tikzpicture}
            
                \filldraw [black] (-2,0) circle (1pt) node[below] {\scriptsize $\text{X}$};
                \filldraw [black] (2,0) circle (1pt) node[below] {\scriptsize $\text{$Z_2$}$};
                \filldraw [black] (0,0) circle (1pt) node[below] {\scriptsize $\text{$Z_1$}$};

                \draw [stealth-, semithick] (0,0) -- (-2,0);
                \draw [stealth-, semithick] (2,0) -- (0,0);

                \draw [dotted, semithick, <->]   (-2,0) to[out=45,in=135] (2,0);

            \end{tikzpicture}
        }         
         \caption{}
         \label{fig:admg5}
     \end{subfigure}
     \begin{subfigure}{0.3\textwidth}
         \centering
        \resizebox{\linewidth}{!}{
            \begin{tikzpicture}
            
                \filldraw [black] (-2,0) circle (1pt) node[below] {\scriptsize $\text{X}$};
                \filldraw [black] (2,0) circle (1pt) node[below] {\scriptsize $\text{$Z_2$}$};

                \draw [dotted, semithick, <->]   (-2,0) to[out=45,in=135] (2,0);

            \end{tikzpicture}
        }         
         \caption{}
         \label{fig:admg6}
     \end{subfigure}
     \caption{}
\end{figure}

\begin{figure}
     \centering
     \begin{subfigure}[b]{0.3\textwidth}
         \centering
        \resizebox{\linewidth}{!}{
            \begin{tikzpicture}
            
                \filldraw [black] (-2,0) circle (1pt) node[below] {\scriptsize $\text{X}$};
                \filldraw [black] (2,0) circle (1pt) node[below] {\scriptsize $\text{$Y$}$};
                \filldraw [black] (0,2) circle (1pt) node[above] {\scriptsize $\text{W}$};
                \filldraw [black] (-1,1) circle (1pt) node[left] {\scriptsize $\text{$R$}$};

                \draw [stealth-, semithick] (2,0) -- (-2,0);
                \draw [stealth-, semithick] (-2,0) -- (-1,1);
                \draw [stealth-, semithick] (-1,1) -- (0,2);

                \draw [dotted, semithick, <->]   (-2,0) to[out=90,in=180] (0,2);
                \draw [dotted, semithick, <->]   (2,0) to[out=90,in=0] (0,2);
                
            \end{tikzpicture}
        }         
         \caption{}
         \label{fig:admg1}
     \end{subfigure}
     \hfill
     \begin{subfigure}[b]{0.3\textwidth}
         \centering
        \resizebox{\linewidth}{!}{
            \begin{tikzpicture}
            
                \filldraw [black] (-2,0) circle (1pt) node[below] {\scriptsize $\text{X}$};
                \filldraw [black] (2,0) circle (1pt) node[below] {\scriptsize $\text{$Y$}$};
                \filldraw [black] (0,2) circle (1pt) node[above] {\scriptsize $\text{W}$};

                \draw [stealth-, semithick] (2,0) -- (-2,0);

                \draw [dotted, semithick, <->]   (-2,0) to[out=90,in=180] (0,2);
                \draw [dotted, semithick, <->]   (2,0) to[out=90,in=0] (0,2);
                
            \end{tikzpicture}
        }         
         \caption{}
         \label{fig:admg2}
     \end{subfigure}
     \begin{subfigure}[b]{0.3\textwidth}
         \centering
        \resizebox{\linewidth}{!}{
            \begin{tikzpicture}
            
                \filldraw [black] (-2,0) circle (1pt) node[below] {\scriptsize $\text{X}$};
                \filldraw [black] (2,0) circle (1pt) node[below] {\scriptsize $\text{$Y$}$};

                \draw [stealth-, semithick] (2,0) -- (-2,0);

            \end{tikzpicture}
        }         
         \caption{}
         \label{fig:admg3}
     \end{subfigure}
     \caption{}
\end{figure}\textbf{}
\section{Efficient Learning of Interventional Distribution}

\subsection{Evaluation}\label{sec:eval}
In this section we give an algorithm for evaluating the interventional probability $P_{\ve{x}}(\ve{y})$ for any $\ve{X}$ of bounded size and assignments $\ve{x},\ve{y}$ to $\ve{X}, \ve{V}\setminus\ve{X}$ respectively. Throughout this section, we will assume $\ve{Y}=\ve{V}\setminus \ve{X}$ and $P_{\ve{x}}(\ve{Y})$ is identifiable. In that case, the following fact is obvious.
\begin{fact}
If $\ve{Y}=\ve{V}\setminus \ve{X}$ and $P_{\ve{x}}(\ve{Y})$ is identifiable, we can remove step 3 and 5a of Algorithm~\ref{algo:id}, and the summations of steps 4, 1, and 5b, without loss of generality. 
\end{fact}


Let $\{\ve{C_i}\}_i$ denote the c-component partition of $G$. Let $\ve{X}=\cup_{i=1}^{\ell} \ve{X_i}$ be a partition of $\ve{X}$ such that $\ve{X_i}\subseteq \ve{C_i}$, for $i=1$ to $\ell$ without loss of generality. We make the following assumption. 
\begin{assumption}
$\pa^{+}(\ve{C_i})$ is $\alpha$-strongly positive for every $1\le i\le \ell$. 
\end{assumption}
In this section, we prove the following main theorem.
\begin{theorem}\label{thm:eval}
We can learn $\wh{P}_{\ve{x}}(\ve{Y})$ using $m=\widetilde{O}\left(\left(\frac{n|\Sigma|^{2k\ell+kd+d}}{ \alpha^{k\ell}\eps^2}+\frac{(3k)^{2(k+3)}|\Sigma|^{2k\ell}}{\alpha^2\eps^2}\right)\log \frac{1}{ \delta}\right)$ samples from $P$ and in $O(m(n+|\Sigma|^{kd+d}))$ time, such that $\tv(\wh{P}_{\ve{x}}(\ve{Y}),P_{\ve{x}}(\ve{Y}))\le \eps$ with high probability.
\end{theorem}

  Define $\ve{C_{>\ell}}:=\cup_{i>\ell} \ve{C_i}$ and $\ve{C_{\le \ell}} :=\cup_{i\le \ell} \ve{C}_i$.  Let $\ve{C_i}\setminus \ve{X_i} = \cup_j \ve{C_{ij}}$ be the c-component partitions of $G[\ve{C_i}\setminus \ve{X_i}]$.  Then any identifiable joint interventional probability factorizes based on the following claim.

\begin{claim}
Let $\ve{Y}=\ve{V}\setminus\ve{X}$, $\ve{An(Y)}=\ve{V}$, and $P_{\ve{x}}(\ve{y})$ is identifiable. Then,
\[P_{\ve{x}}(\ve{y})=P(\ve{c_{>\ell}} \mid \cdo(\ve{c_1},\dots,\ve{c_{\ell}}))\prod_{i,j} P(\ve{c_{ij}}\mid \cdo(\ve{v}\setminus\ve{c_{ij}})).
\] 
\end{claim}
\begin{proof}
Follows from Step~4 of Algorithm~\ref{algo:id}. 
\end{proof}

Hence it suffices for us to learn $R=\prod_{i,j} P(\ve{c_{ij}}\mid \cdo(\ve{v}\setminus\ve{c_{ij}}))=P(\ve{C_{\le \ell}}\setminus \ve{X} \mid \cdo(\ve{x},\ve{c_{>\ell}}))$ and $Q=P_{\ve{c_{\le \ell}}}(\ve{c_{>\ell}})=P(\ve{c_{>\ell}} \mid \cdo(\ve{c_1},\dots,\ve{c_{\ell}}))$.  The former distribution is bounded-dimensional and the later is high-dimensional. So, our approach would be to learn $R$ as an explicit p.m.f. table of size at most $|\Sigma|^k$ and $Q$ as a Bayes net, both up to an additive error at most $\eps$.  

Now we focus on learning the distribution $R=P(\ve{C_{\le \ell}}\setminus \ve{X} \mid \cdo(\ve{x},\ve{c_{>\ell}}))$ for every fixing of $\pa^{-}(\ve{C_{\le \ell}}\setminus \ve{X})$. $R=\prod_{i,j} R_{i,j}$, where each distribution $R_{i,j}=P[\ve{C_{i,j}}\mid \cdo(\ve{v}\setminus \ve{c_{i,j}})]$ starts a recursive call of Algorithm~\ref{algo:id}. Each such recursive call forms a recursion tree whose non-leaves correspond  recursive calls from steps 2, 4, or 5c, and the leaves correspond to recursive calls from steps 1 or 5b. Our strategy would be to give a $(1\pm\eps_2)$-approximate p.m.f. access to the joint distribution for every subsequent recursive call (compared to the true joint distribution), assuming a $(1\pm\eps_1)$-approximate p.m.f. access to the joint distribution of the current call. Inductively, we'll  get a $(1\pm\eps_3)$-approximate access to the leaf distributions: outputs by the leaf calls. Each distribution $R_{i,j}$ is a multiplication of the leaf distributions, possibly with some marginalizations. Marginalization preserves the approximation ratio of the p.m.f. and there are at most $k$ multiplicands, so that our final approximation ratio would be $(1\pm\eps_3)^k$.

Recursive calls are taken from steps 2, 4, or 5c. In steps 4 and 2, the subsequent call just uses the current distribution (or a marginal of it). In step 5c, the subsequent call uses a very different distribution. Therefore, in a path to the leaf in the recursion tree, we need to give a new joint distribution whenever step 5c is taken. We use the following result to give this distribution, based on whether step 5c was taken for the first time in this root-to-leaf path or not. Let $S$ and $D$ be the two distributions just before and during the recursive call at step 7.
\begin{claim}\label{claim:step7Eval} ~\\[-1.5em]
\begin{itemize}
    \item If $S=P$, we can give a $(1\pm 3k\eps)$ p.m.f. of $\wh{D}$ for $D$ using $O((kd+d)\alpha^{-2}\eps^{-2}\log \frac{|\Sigma|}{\delta})$ samples and $O((kd+d)|\Sigma|^{kd+d}\alpha^{-2}\eps^{-2}\log \frac{|\Sigma|}{\delta})$ time.
    \item If $S\neq P$ and we have a $(1\pm\eps)$-approximation $\wh{Q}$ for $Q$,  we can give a $(1\pm 3k\eps)$-approximate p.m.f. $\wh{D}$ for $D$ using $O(k|\Sigma|^{2k})$ time and no samples.
\end{itemize}
\end{claim}
\begin{proof}
Note that each factor of the joint distribution in step 7 is of the form $S[V_i \mid \ve{Z_i}]$, where $\ve{Z_i}$ consists of the effective parents of $V_i$.

In case 1, we take enough samples to empirically learn $S[V_i \mid \ve{z_i}]$ for every fixing of $\ve{z_i}$. Due to $\alpha$-strong positivity, $\min(S[v_i \mid \ve{z_i}],S[\ve{z_i}])\ge \alpha$ for any $v_i,\ve{z_i}$. Hence from Chernoff's bound, using $O(\alpha^{-2}\eps^{-2}\log {1\over \delta})$ samples the learnt distribution would be point-wise $(1\pm\eps)$-close with high probability. Since $|\ve{Z_i}|\le (kd+d)$ our final sample complexity is $m=O((kd+d)\alpha^{-2}\eps^{-2}\log \frac{|\Sigma|}{\delta})$ and the time complexity is $O(m|\Sigma|^{kd+d})$. Since each factor is $(1\pm\eps)$-approximate, the approximation for the joint distribution is at most $(1\pm\eps)^k$.

In case 2, we compute $S[v_i \mid \ve{z_i}]=S[v_i,\ve{z_i}]/Q[\ve{z_i}]$ by appropriate marginalizations of $S$, which preserves the approximation. Since step 7 is taken at least once before, the graph-size is at most $k$. Due to the ratio along with a maximum of $k$ multiplications, the final approximation becomes $\left(\frac{1+\eps}{1-\eps}\right)^k$-factor. This does not involves sampling and can be done in $O(k|\Sigma|^{2k})$ time, since in this case $Q$ must be over at most $k$ variables.

In either case, we return $\wh{D}$ as a p.m.f. table of size $|\Sigma|^k$.
\end{proof}

Similarly for the leaf calls from step 1 or 5b, we approximate their (terminal) output distributions (henceforth referred to as leaf distributions) by marginalization or conditional sampling, depending on whether the current distribution is $P$ or not. Inductively, we would get an approximation guarantee between the true and estimated leaf distributions. Our output just consists of p.m.f. tables for each joint distribution $\widehat{S}$ just before a leaf call. In that case, each leaf distribution $\widehat{D}$ is simply a marginalization (step 1) or product of conditional probabilities (step 6) of $\widehat{S}$.

We now consider the leaf calls from step 1 or 5b. Again, we split into two cases, depending on whether the leaf call was taken using the original distribution $P$ or not. Let $D$ be the distribution of the leaf call.
\begin{claim}
Let $D$ be a true leaf distribution in the recursion tree. Then our corresponding distribution $\wh{D}$ as mentioned above is point-wise $(1\pm(3k)^k\eps)$-approximate for $D$.
\end{claim}
\begin{proof}
If $D=P$ we get a sampling access to $P$. If $D\neq P$, we get $D$ as a p.m.f. table. 

Step 1 returns the joint distribution over $\ve{Y}$, over at most $|\Sigma|^k$ items. If $D=P$, we learn this distribution up to point-wise $(1\pm\eps)$-factor using $\wt{O}(k\alpha^{-1}\eps^{-2})$ samples with high probability. If $D\neq P$, $D$ itself would be the output distribution. In either case, we return a p.m.f. table of size at most $|\Sigma|^k$.

If step 5c is taken with $D=P$, we need learn each $P[v_i \mid \ve{z_i}]$, where $\ve{Z_i}$ is the effective parent set of $V_i$, and is of size at most $(kd+d)$. So, we iterate through all possible $\ve{z_i}$ and learn each term $P[v_i \mid \ve{z_i}]$ point-wise up to $(1\pm\eps)$-error by rejection sampling with high probability. There are at most $k$ such terms, so that the sample complexity would be $m=\wt{O}((kd+d)\alpha^{-2}\eps^{-2})$ and the time complexity would be $O(m|\Sigma|^{kd+d})$.

If step 5c is taken with $D\neq P$, then it must be preceded by a call from step 5b, since only that step changes the distribution by our construction. Then $D$ is a distribution over at most $k$ variables. We obtain all possible $D[v_i \mid \ve{z_i}]= D[v_i \circ \ve{z_i}]/D[\ve{z_i}]$ terms by appropriate marginalizations of $D$ using $\wt{O}(|\Sigma|^{2k})$ time and no samples. 

In both the above cases of step 5c, we return all possible conditional probability tables required for evaluating the output formula of this step.
Note that the recursion depth is at most $k$ and we lose a factor of $(1\pm3k\eps)$ from $(1\pm\eps)$ in each depth. Therefore, the approximation ratio for the leaf distributions is at most $(1\pm (3k)^k\eps)$.
\end{proof}

Each $R_{i,j}$ is a product of several leaf distributions. Our final output distribution $\wh{R}_{i,j}$ just uses $\wh{S}$ in place of each such leaf distribution $S$. Since $R=\prod_{i,j} R_{i,j}$ is a product of at most $k\ell$ leaf distributions, the following claim is obvious.
\begin{claim}\label{claim:rij}
$\wh{R}$ is point-wise $(1\pm (3k)^{k+1}\ell\eps)$-approximate for $R$.
\end{claim}

We next focus on learning the distribution $Q$.
\begin{lemma}\label{lem:large}
Let $\ve{\Pa^+(C_i)}$ be $\alpha$-strongly positive. Then for every fixing of $\ve{C_{\le \ell}}$, $Q$ can be learnt as a Bayes net $\widehat{Q}$ of indegree at most $(kd+d)$ using $m=\widetilde{\Theta}\left(\frac{n|\Sigma|^{kd+d}}{ \alpha^{|\ve{C_{\le \ell}}|}\eps}\log {n |\Sigma|^{kd+d+1}\over \delta}\right)$ samples and $O(mn)$ time such that $\kl(Q,\widehat{Q})\le \eps$ with probability at least $(1-\delta)$.
\end{lemma}
We will now introduce some of the tools used in proving Lemma~\ref{lem:large}.  The following Claim relates the p.m.f.s of $Q$ and $P$.  Here we used Tian's factorization to write $Q=\prod_{V_i\in \ve{C_{>\ell}}} P(V_i\mid \ve{Z_i})$, where $\ve{Z_i}$ is the effective parents of $V_i$.

\begin{claim}\label{claim:relate}
Let $\ve{v}$ be an assignment to $\ve{V}$ and $\ve{w}$ be its restriction to $\ve{C_{>\ell}}$. Then $\alpha^{|\ve{C_{\le \ell}}|}\le \frac{P[\ve{v}]}{Q[\ve{w}]}\le 1$.
\end{claim}

\begin{proof}
$\frac{P[\ve{v}]}{Q[\ve{w}]}= \frac{
\prod_{V_i\in \ve{V}} P(v_i\mid \ve{z_i})
}{
\prod_{V_i\in \ve{C_{>\ell}}} P(v_i\mid \ve{z_i})
}
= \prod_{V_i \in \ve{C_{\le \ell}}} P(v_i\mid \ve{z_i})
\ge \prod_{V_i \in \ve{C_{\le \ell}}} P(v_i\circ \ve{z_i})
\ge \alpha^{|\ve{C_{\le \ell}}|}.$
\end{proof}
We closely follow the $\kl$-learning result for Bayes nets given in~\cite{chowliu}. Let $\ve{Z_i'}=\ve{Z_i}\setminus \ve{C_{\le \ell}}$.  Our algorithm just learns the add-1 empirical distribution $\wh{P}[v_i \mid \ve{Z_i'}=\ve{a}])$ on the conditional samples from $P[v_i \mid \ve{Z_i'}=\ve{a}]$. Our learnt distribution $\wh{Q}$ consists of the $\wh{P}[v_i \mid \ve{Z_i'}=\ve{a}])$'s, in place of every $P[v_i \mid \ve{Z_i'}=\ve{a}])$ in the Bayes net factorization for $Q$.
\begin{fact}[\cite{dasgupta, cdks, chowliu}]\label{fact:kllocal}
$$\kl(Q,\wh{Q})=\sum_{i\in \ve{C_{>\ell}}}\sum_{\ve{a}} Q[\ve{Z_i'}=\ve{a}]\kl(P[v_i \mid \ve{Z_i'}=\ve{a}],\wh{P}[v_i \mid \ve{Z_i'}=\ve{a}]).$$
\end{fact}
We also have the following guarantee about the add-1 estimator.
\begin{fact}[\cite{chowliu}]\label{fact:add1}
Let $D$ be an unknown distribution over $k$ items and $\wh{D}$ be its add-1 empirical distribution of $m$ samples. Then if $m\gtrsim {k\over \eps}\log {k\over \delta}\left( \log {k\over \eps} + \log \log {k\over \delta}\right)$ then $\kl(D,\wh{D})\le \eps$ with probability at least $(1-\delta)$.
\end{fact}
We are now ready to prove Lemma~\ref{lem:large}.
\begin{proof}[Proof of Lemma~\ref{lem:large}] We analyze the two cases: $Q[\ve{Z_i'}=\ve{a}]\ge \frac{\eps}{n|\Sigma|^{kd+d}\log(m+|\Sigma|)}$ and otherwise. 

In the former case, $P[\ve{Z_i'}=\ve{a}]\ge\frac{ \alpha^{|\ve{C_{\le \ell}}|}\eps}{n|\Sigma|^{kd+d+1}}$ from Claim~\ref{claim:relate} and $m= \widetilde{\Theta}\left(\frac{n|\Sigma|^{kd+d}}{ \alpha^{|\ve{C_{\le \ell}}|}\eps}\allowbreak \log {n |\Sigma|^{kd+d+1}\over \delta}\right)$ samples would ensure at least $\widetilde{\Theta}(\frac{n|\Sigma|^{kd+d+1}\cdot Q[\ve{Z_i'}=\ve{a}]}{\eps}\log {n |\Sigma|^{kd+d+1}\over \delta})$ conditional samples are seen from $P[v_i \mid \ve{Z_i'}=\ve{a}]$ with high probability from Chernoff's bound. Hence, $\kl(P[v_i \mid \ve{Z_i'}=\ve{a}],\wh{P}[v_i \mid \ve{Z_i'}=\ve{a}])\le \frac{\eps}{n|\Sigma|^{kd+d}Q[\ve{Z_i'}=\ve{a}]}$ from Fact~\ref{fact:add1} except with probability at most $\frac{\eps}{n|\Sigma|^{kd+d}}$.

In the later case, $Q[\ve{Z_i'}=\ve{a}]\le \frac{\eps}{n|\Sigma|^{kd+d}\log(m+|\Sigma|)}$ and $\kl(P[v_i \mid \ve{Z_i'}=\ve{a}],\wh{P}[v_i \mid \ve{Z_i'}=\ve{a}])\le \log(m+|\Sigma|)$, since the least add-1 probability of at most $m$ conditional samples is $\frac{1}{m+|\Sigma|}$.

Combining all the cases, the summation of the RHS of Fact~\ref{fact:kllocal} evaluates to at most $O(\eps)$.
\end{proof}

Now we describe the overall construction of the evaluator for $P_{\ve{x}}(\ve{Y})=Q\cdot R$, where $Q=\prod_{V_i \in \ve{C_{>\ell}}}Q(V_i \mid \ve{Z_i})$, $R=\prod_{V_i \in \ve{C_{\le \ell}}} S_i(V_i \mid \ve{Z_i})$, $\ve{Z_i}$ is the conditioning set for $V_i$, $Q$ is the Bayes net and $S_i$'s are the leaf distributions defined before. This is very similar to a Bayes net factorization except that the probability distribution for the factors need not be the same. Our evaluator is 
\begin{align}
\wh{P}_{\ve{x}}(\ve{y})=\prod_{V_i \in \ve{C_{>\ell}}}\wh{Q}(V_i \mid \ve{Z_i}) \prod_{V_i \in \ve{C_{\le \ell}}} \wh{S}_i(V_i \mid \ve{Z_i}),\label{eqn:phat}
\end{align}
where $\widehat{Q}$ comes from Lemma~\ref{lem:large} and $\wh{S}_i$ comes from $R$ in Claim~\ref{claim:rij}.

\begin{proof}[Proof of Theorem~\ref{thm:eval}]
We get $\tv(\prod_{V_i \in \ve{C_{>\ell}}}Q(V_i \mid \ve{Z_i}),\prod_{V_i \in \ve{C_{>\ell}}}\widehat{Q}(V_i \mid \ve{Z_i}))\le \eps$ for any fixed $\ve{c_{\le \ell}}$ from Lemma~\ref{lem:large} and Pinsker's inequality, using $m=\widetilde{\Theta}\left(\frac{n|\Sigma|^{kd+d}}{ \alpha^{|\ve{C_{\le \ell}}|}\eps^2}\log {n |\Sigma|^{kd+d+1}\over \delta}\right)$ samples and $O(mn)$ time with high probability.

We scale down $\eps$ by a $(3k)^{k+1}\cdot \ell$ factor to get that for any fixed $\ve{c_{>\ell}}$, $\wh{R}$ is point-wise $(1\pm\eps)$-approximate for $R$ with high probability using $m=O((3k)^{2(k+3)}\ell^3d\alpha^{-2}\eps^{-2}\log {|\Sigma|k\ell\over \delta})$ samples and $O(m|\Sigma|^{kd+d})$ time from Claim~\ref{claim:step7Eval} and Claim~\ref{claim:rij}.

Combining the above two pieces, we get at most  $ \eps(|\Sigma|^{k\ell}+1)$ error as follows. We get the theorem by an appropriate scaling.
\begin{align*}
    &\tv(P_{\ve{x}}(\ve{y}),\wh{P}_{\ve{x}}(\ve{y}))\\
    &=\sum_{\ve{c_{\le \ell}},\ve{c_{>\ell}}} \left|
    \prod_{V_i \in \ve{C_{>\ell}}}Q(V_i \mid \ve{Z_i}) \prod_{V_i \in \ve{C_{\le \ell}}} S_i(V_i \mid \ve{Z_i})
    -
    \prod_{V_i \in \ve{C_{>\ell}}}\wh{Q}(V_i \mid \ve{Z_i}) \prod_{V_i \in \ve{C_{\le \ell}}} \wh{S}_i(V_i \mid \ve{Z_i})
    \right|\\
    &= \sum_{\ve{c_{\le \ell}},\ve{c_{>\ell}}} \left|
\prod_{V_i \in \ve{C_{>\ell}}}Q(V_i \mid \ve{Z_i}) \prod_{V_i \in \ve{C_{\le \ell}}} S_i(V_i \mid \ve{Z_i})
    -\prod_{V_i \in \ve{C_{>\ell}}}Q(V_i \mid \ve{Z_i}) \prod_{V_i \in \ve{C_{\le \ell}}} \wh{S}_i(V_i \mid \ve{Z_i})\right.\\
    &\qquad\qquad\left.
    +\prod_{V_i \in \ve{C_{>\ell}}}Q(V_i \mid \ve{Z_i}) \prod_{V_i \in \ve{C_{\le \ell}}} \wh{S}_i(V_i \mid \ve{Z_i})
    -
    \prod_{V_i \in \ve{C_{>\ell}}}\wh{Q}(V_i \mid \ve{Z_i}) \prod_{V_i \in \ve{C_{\le \ell}}} \wh{S}_i(V_i \mid \ve{Z_i})
        \right|\\
    &\le \sum_{\ve{c_{>\ell}}}\sum_{\ve{c_{\le \ell}}}\left|\prod_{V_i \in \ve{C_{>\ell}}}Q(V_i \mid \ve{Z_i})\prod_{V_i\in \ve{C_{\le\ell}}}S_i(V_i \mid \ve{Z_i})-(1\pm\eps)\prod_{V_i \in \ve{C_{>\ell}}}Q(V_i \mid \ve{Z_i})\prod_{V_i\in \ve{C_{\le \ell}}}S_i(V_i \mid \ve{Z_i})\right|\\
    &\qquad\qquad\qquad\qquad+
    \sum_{\ve{c_{\le\ell}}}\sum_{\ve{c_{> \ell}}}\left|\prod_{V_i\in \ve{C_{>\ell}}} Q_i(V_i \mid \ve{Z_i})-\prod_{V_i\in \ve{C_{>\ell}}}\wh{Q}_i(V_i \mid \ve{Z_i})\right|\\
    &\le \eps\sum_{\ve{c_{>\ell}}}\sum_{\ve{c_{\le \ell}}}\prod_{V_i \in \ve{C_{>\ell}}}Q(V_i \mid \ve{Z_i})\prod_{V_i \in \ve{C_{\le \ell}}} S_i(V_i \mid \ve{Z_i})  + \sum_{\ve{c_{\le\ell}}}\eps
    \le \eps(1+|\Sigma|^{k\ell}).\qedhere
\end{align*}
\end{proof}
\subsection{Generator}\label{sec:samp}
In this section, we give an algorithm for generating samples approximately according to the distribution $P_{\ve{x}}(\ve{y})$. Our algorithm generates samples according to the distribution $\wh{P}_{\ve{x}}(\ve{y})$ given in \eqref{eqn:phat}.
\begin{theorem}
We can generate independent samples from $\wh{P}_{\ve{x}}(\ve{y})$ in $O(n)$ time.
\end{theorem}
\begin{proof}
We have $P_{\ve{x}}(\ve{y})=\prod_{V_i \in \ve{C_{>\ell}}}\wh{Q}(V_i \mid \ve{Z_i}) \prod_{V_i \in \ve{C_{\le \ell}}} \wh{S}_i(V_i \mid \ve{Z_i})$ from \eqref{eqn:phat}, where $\hat{Q}$ comes from Lemma~\ref{lem:large} and $\wh{S}$ comes from $R$ of Claim~\ref{claim:rij}. Note that either in Lemma 2 or Claim 8, the effective graph in any step of the recursion is always a subgraph of $G$. Hence, the topological order between $V_i$ and $\ve{Z_i}$ is never violated. In particular, we can sample each random variable of $\ve{Y}$ in the topological order of $\ve{V}\setminus \ve{X}$. Then at any step of this sampling, whenever we try to sample some $v_i\sim \wh{Q}(V_i \mid \ve{Z_i})$ or $\sim \wh{S}_i(V_i\mid \ve{Z_i})$, $\ve{Z_i}$ would always be sampled before it and assigned.
\end{proof}

\begin{remark}
\emph{
It follows that we can also sample from any $\ve{T}\subseteq \ve{V}\setminus \ve{X}$, ignoring the unnecessary variables. However, evaluating the marginal on $\ve{T}$ in general remains computationally expensive.
}
\end{remark}
\subsection{Examples}
In this section, we illustrate our evaluator and generator with the help of two simple examples.

\paragraph{Example 1:} Consider the graph in Figure~\ref{fig:admg4} taken from \cite{tianthesis}. We would like to learn the intervention $P_{x}(z_1,z_2,y)$. The following sequence of steps are taken in our algorithm.
\begin{enumerate}
    \item The starting graph has two c-components $(X,Z_2)$ and $(Y,Z_1)$. Only $(X,Z_2)$ contains an intervening variable. So, we'll assume $(X,Z_1,Z_2)$ is $\alpha$-strongly positive.
    \item Step 4 is taken first and generates: ID($(y,z_1),(x,z_2),P,G$) $*$ ID($z_2,(x,y,z_1),P,G$). Since, $P$ is unchanged, we don't do anything.
    \item The first ID takes step 6 and generates the formula: $P[y \mid x,z_1,z_2]P[z_1\mid x]$. Note that we are learning this distribution jointly as a Bayes net, assuming $\alpha$-strong positivity of only $(X,Z_1,Z_2)$.
    \item $Y$ is marginalized out in step 2 from the second ID. Hence, we don't change the distribution. Next call is ID($z_2,(x,z_1),P[X,Z_1,Z_2],H$), where $H$ is the graph of Figure~\ref{fig:admg5}.
    \item Step 7 is taken next: ID($z_2,x,S[X',Z_2],T$), where $S[X',Z_2]=P[X',Z_2 \mid do(z_1)]=P[X']P[Z_2\mid z_1,X']$. We learn $S$ up to point-wise $(1\pm 6\eps)$-factor using $\wt{O}(\alpha^{-2}\eps^{-2})$ samples and store it as a table. $T$ is the graph shown in Figre~\ref{fig:admg6}.
    \item Step 2 next prunes $X$ from $T$. The distribution is a marginal and hence $S$ is unchanged.
    \item Step 1 generates the formula: $S[z_2]$.
\end{enumerate}
At this point, we will be able to evaluate and sample $P[y \mid x,z_1,z_2]P[z_1\mid x]S[z_2]$ using our Bayes net $P$ and the explicit table for $S$.

\paragraph{Example 2:} Consider the graph in Figure~\ref{fig:admg1} taken from \cite{bareinboim}. We are interested to learn the intervention $P_{(w,r,y)}(Y)$. The following sequence of steps are taken in our algorithm.
\begin{enumerate}
    \item The starting graph $G$ has two c-components: $(W,X,Y)$ and $R$. Both these components contain an intervening variable and hence assumed $\alpha$-strongly positive together with their parents.
    \item First, step 7 is taken: ID($y,(w,x),Q[W',X',Y],H$), where $Q[W',X',Y]=P[W'X'Y \mid \cdo(r)]=P[W']P[X' \mid W',r]P[Y \mid X',r,W']$. So, we learn $Q$ as a point-wise $(1\pm9\eps)$-apporox. p.m.f. table using $\wt{O}(\alpha^{-2}\eps^{-2})$ samples from $P$. The graph $H$ is shown in Figure~\ref{fig:admg2}.
    \item Next, step 2 prunes $W$ from $H$. The new distribution is the marginal of $Q$ on $X',Y$ and so $Q$'s table is passed on unchanged. The new graph $T$ is shown in Figure~\ref{fig:admg3}. So, the next call is ID($y,x,Q[X',Y],T$).
    \item Finally, step 6 is taken, which generates the final formula $Q[y \mid x]$.
\end{enumerate}
At this point, we will be able to evaluate and sample $Q[y \mid x]$ using our stored table for $Q$.

\section{Hardness of Evaluating Marginals}\label{sec:hard}

In this section we show the computational hardness for approximately giving a succinct evaluator for the {\em marginal} of an interventional distribution in general. This is in contrast to our result in Section~\ref{sec:eval}, where we showed for the specific case that whenever $\ve{X}$ is of bounded size and $\ve{Y}=\ve{V}\setminus \ve{X}$, we can efficiently give a generator as well as an evaluator for $\wh{P}_{\ve{x}}(\ve{Y})$, which was $\eps$ close to $P_{\ve{x}}(\ve{Y})$ in $\tv$. Thus, our hardness in this section applies to the case when $\ve{Y}\subset \ve{V}\setminus \ve{X}$. In fact we show the hardness for the empty intervention $P(\ve{Y})$, which is just the marginal of a completely observed Bayesian network of indegree 2.

We need the following definitions.

\begin{definition}[Polynomial-time samplable distributions]
Given a Boolean circuit $C_n$ mapping  $n$ bits to $m$ bits, the distribution sampled by $C_n$ is obtained by uniformly choosing $x \in  \{0,1\}^n$ and evaluating $C$ on $x$. A distribution is polynomial-time samplable distribution if it is sampled by a circuit $C_n$ of size $\poly(n)$. We often use $C$ itself to denote the distribution sampled by the circuit $C_n$.
\end{definition}

We will use the following hardness result for the tolerant testing of two polynomial-time samplable distributions.

\begin{definition}[Tolerant testing of polynomial-time samplable distributions]
Let $\distckt$ be the following problem: given the encodings of two boolean circuits $C_n$ and $D_n$, both of which output exactly $m=\poly(n)$ bits, distinguish between the two cases: $\tv(C_n,D_n)\le 1/3$ versus $\tv(C_n,D_n)\ge 2/3$. 
\end{definition}
\begin{theorem}[\cite{SV03}]
$\distckt$ is complete for the complexity class Statistical Zero Knowledge (denoted $\szk$).

\end{theorem}
The class $\szk$ contains several hard computational problems including the Graph Isomorphism problem and is believed to be 
computationally harder than $\bpp$; the class of problems that admit efficient randomized algorithms.

In the following, we show that $\distckt$ reduces to the problem of computing an efficient, approximate evaluator for the marginal distribution of a Bayes net, thereby showing that the later problem is hard. We formally define this problem now.

\begin{definition}[Efficient $\eps$-approximate evaluator circuit]
A circuit $\mathcal{E}_P: \{0,1\}^n\rightarrow \mathbb{R}$ is called an efficient $\eps$-additive evaluator for a distribution $P$ over $\{0,1\}^n$, if there exists a distribution $\wh{P}$ over $\{0,1\}^n$ such that:
\begin{enumerate}
    \item {[Additive approximation]} $\tv(P,\wh{P})\le \eps$,
    \item {[Evaluation query]} For any $x\in \{0,1\}^n$, $\mathcal{E}_P$ on input $x$ outputs a number $p\in[1-\eps,1+\eps]\wh{P}(x)$ in $\poly(n)$ time,
\end{enumerate} 
\end{definition}

\begin{definition}
Let $\bayesmar$ be the following problem: given a parameter $\eps$, samples from an unknown Bayes net $P$, and a $S\subseteq [n]$, output an efficient $\eps$-approximate evaluator circuit $Q:\{0,1\}^{|S|}\rightarrow \mathbb{R}$ for the marginal distribution of $P$ over $S$.
\end{definition}

We also need the following result for approximating the variation distance additively.
\begin{theorem}[\cite{bgmv20}]\label{thm:bgmv}
Let $P$ and $Q$ be two unknown distributions over $\Omega$. Then given access to two efficient $\eps$-approximate evaluator circuits $\mathcal{E}_P$ and $\mathcal{E}_Q$ for $P$ and $Q$, we can estimate $\tv(P,Q)$ up to an additive $4\eps$ error with $(1-\delta)$ probability using $O(\eps^{-2}\log {1\over \delta})$ samples from $P$ and $O(\eps^{-2}\log {1\over \delta})$ evaluation queries to $\mathcal{E}_P$ and $\mathcal{E}_Q$.
\end{theorem}

\begin{theorem}\label{thm:bayesmar}
If $\bayesmar$ has a randomized polynomial-time algorithm even for Bayes nets of indegree at most 2, then $\distckt\in \bpp$ and hence $\szk \subseteq \bpp$.
\end{theorem}
\begin{proof}
Let $\mathcal{A}$ be the hypothetical randomized polynomial-time algorithm for $\bayesmar$. We solve $\distckt$ on $C_n$ and $D_n$ using $\mathcal{A}$ as follows.

Firstly, we assume without loss of generality all the AND, OR, NOT gates of $C_n$ and $D_n$ have at most 2-arguments, since otherwise we stack such gates, increasing the circuit size by at most a polynomial. Then, the two circuits are exactly Bayes nets of indegree at most 2, whose source nodes are random and intermediate nodes follow deterministic functions. We also denote by $C_n$ and $D_n$ these two samplable distributions with a slight abuse of notation. Let $S_C$ and $S_D$ denote the sets of output bits of $C_n$ and $D_n$, such that $|S_C|=|S_D|$. Let $C_{\all}$ and $D_{\all}$ denote the joint distributions of {\em all the gates} of $C_n$ and $D_n$. Hence, $C_n$ and $D_n$ are respectively the marginals of $C_{\all}$ and $D_{\all}$ over the sets $S_C$ and $S_D$.  

We run $\mathcal{A}$ on $(C_{\all},S_C,1/40)$ to get an efficient $\eps$-approximate evaluator $\mathcal{E}_C$ for $C_n$ with high probability. Similarly we get $\mathcal{E}_D$ for $D_n$. Moreover, $C_n$ can be sampled in randomized polynomial time. Hence from Theorem~\ref{thm:bgmv}, using $\mathcal{E}_C$, $\mathcal{E}_D$, and samples from $C_n$, we can approximate $\tv(C_n,D_n)$ up to an additive error $1/10$ in polynomial-time with high probability. This would show $\distckt\in\bpp$. 
\end{proof}

Finally, we show a stronger hardness result for getting an evaluator for the marginal of a Bayes net if we want a {\em multiplicative} approximation of the p.m.f.,  formally defined below.
\begin{definition}[Efficient $c$-multiplicative evaluator circuit]
A circuit $\mathcal{E}_P:\{0,1\}^n\rightarrow \mathbb{R}$ is called an efficient $c$-multiplicative evaluator circuit for a distribution $P$ over $\{0,1\}^n$, if there exists a distribution $\wh{P}$ over $\{0,1\}^n$ such that:
\begin{enumerate}
    \item {[Multiplicative approximation]} $\wh{P}(x)/P(x)\in[1/c,c]$ for some constant $c>1$ and for any $x\in {0,1}^n$.
    \item {[Evaluation]} for any $x\in \{0,1\}^n$, $\mathcal{E}_P$ on input $x$ outputs $\wh{P}(x)$ in $\poly(n)$ time.
\end{enumerate}
\end{definition}

\begin{definition}[\bayesmarmult]
Let $\bayesmarmult$ be the following problem: given a Bayesian network $P$ over the binary variables $[n]$, and a $S\subseteq [n]$, return an efficient $c$-multiplicative evaluator circuit for the marginal distribution of $P$ over $S$.
\end{definition}

We reduce from the circuit evaluation problem which is well-known to be $\np$-hard.
\begin{definition}[\circeval]
Given the encoding of a Boolean circuit $C:\{0,1\}^n\rightarrow\{0,1\}$ as input, decide whether there exists an $x$ such that $C(x)=1$ or not.
\end{definition}

\begin{theorem}
$\circeval$ is $\np$-complete.
\end{theorem}

Now, we give a reduction from $\circeval$ to $\bayesmarmult$, showing that the later problem is unlikely to be in randomized polynomial-time.

\begin{theorem}
If $\bayesmarmult$ has a randomized polynomial time algorithm even for Bayes nets of indegree at most 2, then $\circeval\in \bpp$ and hence $\np\subseteq \bpp$.
\end{theorem}
\begin{proof}
Let $\mathcal{A}$ be a hypothetical randomized polynomial time algorithm for $\bayesmarmult$. Let $C:\{0,1\}^n\rightarrow \{0,1\}$ be the instance of the $\circeval$ problem. Let $b$ denote the output bit of $C$. We also denote by $C$, the joint distribution of all its gates when its input bits are chosen randomly. 

As argued in the proof of Theorem~\ref{thm:bayesmar}, $C$ is a Bayes net of indegree at most 2 over all its gates without loss of generality.  We invoke $\mathcal{A}$ with the subset being $\{b\}$ and the mulplicative ratio being any constant $c>1$. If $C$ has no satisfying input, then $b\sim \ber(0)$. If $C$ has a satisfying input, $b\sim \ber(p)$ for some $p\ge 1/2^n$. Therefore, a $c$-factor approximation of the bias of $b$ would decide $\circeval$. 
\end{proof}

\paragraph{Acknowledgement} A.B.~was partially supported by an NRF Fellowship for AI (R-252-000-B13-281), a startup grant (R-252-000-A33-133), and an Amazon Research Award (R-252-000-A61-720). S.G.~was supported by A.B.'s NRF Fellowship for AI (R-252-000-B13-281). S.K.~was supported by NSF 1846300 CAREER and NSF 1815893 grants. N.V.V.~was supported in part by National Science Foundation HDR:TRIPODS-1934884.

\bibliographystyle{abbrvnat}
\bibliography{ref}

\end{document}